\documentclass[a4paper]{scrartcl}
\usepackage{amssymb}
\usepackage{amsmath}
\usepackage{amsthm}
\usepackage{mathabx}
\usepackage{latexsym}
\usepackage{enumerate}
\usepackage{graphicx,color}
\usepackage{hyperref}
\usepackage{dsfont}
\usepackage{xcolor}
\usepackage{enumitem}

\usepackage[backend=bibtex,bibencoding=ascii,style=authoryear-comp,natbib=true,dashed=false,maxcitenames=2]{biblatex} 

\renewcommand{\d}{\mathrm{d}}
\newcommand{\E}{\mathbb{E}}
\DeclareMathOperator*{\plim}{\mathrm{plim}}

\numberwithin{equation}{section}

\numberwithin{figure}{section}

\theoremstyle{plain}
\newtheorem{theorem}{Theorem}[section]
\newtheorem{proposition}[theorem]{Proposition}

\theoremstyle{definition}
\newtheorem{definition}[theorem]{Definition}

\newtheorem{example}[theorem]{Example}

\usepackage{authblk}

\title{On the decomposition of an insurer's profits and losses}

\author[1]{Marcus C.~Christiansen}

\affil[1]{\footnotesize Institut f{\"u}r Mathematik, Carl von Ossietzky Universit{\"a}t Oldenburg,  Carl-von-Ossietzky-Stra{\ss}e 9--11, DE-26129 Oldenburg, Germany.}

\date{\today}

\begin{document}

\maketitle


\begin{abstract}
 Current reporting standards for insurers require a decomposition of observed profits and losses in such a way that changes in the insurer's balance sheet can be attributed to specified risk factors.
  Generating such a decomposition is  a nontrivial task because  balance sheets generally depend on the risk factors in a non-linear way. This paper starts from an axiomatic perspective on profit and loss decompositions and finds that the axioms  necessarily lead to infinitesimal sequential updating (ISU) decompositions, provided that the latter exist and are stable, whereas  the current practice is rather to use sequential updating (SU) decompositions.    The generality of the axiomatic approach makes the results useful also beyond insurance applications wherever profits and losses shall be additively decomposed in a risk-oriented manner.
 \end{abstract}

Keywords: IFRS 17; Solvency II; change analysis; surplus decomposition;  conditional risk measure

\section{Introduction}
Profits and losses that emerge in an insurer's balance sheet between two reporting dates can stem from various sources. The  international financial reporting standard IFRS 17 of the  International Accounting Standards Board (2017),  the MCEV reporting principles of the CFO Forum (2016) as well as the insurance regulation of the European Union (2015) require a change analysis of the insurer's own funds (or solvency reserves)
that identifies and quantifies the sources of the observed profits and losses.
Alternative names for  change analysis are profit and loss attribution,  analysis of movement, or variation analysis. The implementation of these guidelines by insurance companies is a non-trivial task, cf.~Candland \& Lotz (2014) and Bashford \& Dubischar (2020).

This paper starts from an axiomatic perspective on profit and loss decompositions,  calling for additivity, normalization and stability of decompositions. Additive decompositions are easy to interpret and allow the insurer to redistribute profits and losses among different parties.  Normalization means that no profits and losses are attributed to a risk factor that had no updates in the last reporting period. Stability describes the property of a decomposition to be continuous with respect to time lags or time quantization in the empirical risk factor data. 
Stability ensures that these perturbations in the empirical observations  do not change a decomposition too much.

We take a time-dynamic perspective here where profit and loss decomposition means that we are decomposing a discounted surplus process into a sum of partial discounted surplus processes that each  uniquely refer to specified time-dynamic risk factors. In the special case of a one-period model, the problem simplifies to decomposing a surplus random variable to a sum of partial surplus random variables. The latter kind of problem is discussed at length in Schilling et al.~(2020), who introduce the notion of  `meaningful risk decompositions' by a list of six axioms. Schilling et al.~(2020) show that none of the decompositions suggested in the literature before are meaningful risk decompositions and come up with a new and meaningful decomposition principle that bases on the martingale representation theorem (MRT).  While their axioms are time-static in nature, the MRT decomposition is actually a time-dynamic decomposition of a discounted surplus process. Unfortunately, the MRT decomposition concept is limited to cases where the insurer's discounted surplus process is a martingale. So, unless the insurer's balance sheet valuations are purely based on conditional expectations, the MRT decomposition principle is not applicable for  profit and loss attribution.

Among practitioners, a popular method for profit and loss attribution is a sequential updating of the insurer's balance sheet by iterating the balance sheet calculations and  layering the changes of each risk factor. This sequential updating (SU) decomposition principle produces additive decompositions and is frequently used in the economics literature, but unfortunately it depends on the update order of the risk factors, cf.~Fortin et al.~(2011) and Biewen (2014). The so-called one-at-a-time method avoids that problem by resetting after each update step all parameters to previous year's status, but the resulting decomposition is not additive anymore, cf.~Shorrocks (2013) and Schilling et al.~(2020). In the context of surplus redistribution in with-profit life insurance, Jetses \& Christiansen (2021) recently introduced an infinitesimal sequential updating (ISU) decomposition principle that eliminates the disadvantages of  sequential updating by decreasing the length of the reporting periods down to zero so that the impact of the update order vanishes. More specifically,
 ISU decompositions are defined as limits of SU decompositions for  update interval lengths going to zero. While the asymptotic theory can be technically challenging, in insurance practice ISU decompositions are   easy to implement  since by construction they can be simply approximated by SU decompositions.

The central finding of this paper is that our axiomatic approach necessarily leads to  ISU decompositions, provided that the latter exist and are stable.  So, this paper finally promotes the same decomposition principle as Jetses \& Christiansen (2021), but motivation and justification are different.
We look at three examples that illustrate the results of this paper. The examples are kept simple in order to make the fundamental mechanism easily comprehensible. The analysis of more complex examples is left for future research.

The paper is structured as follows.
The second section defines the profit and loss attribution problem and sets the mathematical frame for profit and loss decompositions. The third section introduces the axiomatic perspective. In the fourth section we recall the SU and ISU decomposition principles. The key message of the paper is in section five, where we show that our axioms lead to ISU decompositions.
 Section six calculates a couple of examples, where the technical proofs are outsourced to section seven. The eighth and last section draws a conclusion.

\section{The insurer's surplus process}\label{SectionSurplusProcess}

This section bases on Christiansen \& Jetses (2021) but takes a more general perspective beyond surplus decompositions in with-profit life insurance.
  We generally assume that we have a complete probability space $(\Omega , \mathcal{A},\mathbb{P})$   with a 
  right-continuous, complete filtration $\mathbb{F}=(\mathcal{F}_t)_{t \geq 0}$.
Let $\d C (t)  $ be the insurer's cash flow at time $t$.
We use the convention that outgoing payments get a negative sign and incoming payments get a  positive sign, looking from the perspective of the insurer. Let $\d \Phi (t) $ be the return on investment of the insurer's investment portfolio at time $t$.
Assuming that  $C$ is a finite variation semimartingale,
the asset value process $A$ of the insurer is a semimartingale with the dynamics
\begin{align*}
  \d A(t)  =  A(t-) \, \d \Phi(t) + \d C(t).
\end{align*}
By defining  $\kappa$  as the value process of a self-financing portfolio with investment return $\Phi$ and with a starting value of one, i.e.~$\kappa$ is the solution of
\begin{align*}
  \d \kappa (t) = \kappa(t-) \, \d \Phi(t),\quad \kappa(0)=1,
\end{align*}
 the solution of the stochastic differential equation of the asset value process can be represented as
\begin{align*}
  A(t) =  \kappa(t) \, A(0)    +  \int_{(0,t]} \frac{\kappa(t)}{\kappa(s)} \, \d C(s),
\end{align*}
provided that $\kappa$ is strictly positive.
In order to see that, apply integration by parts on the latter formula.
In the hypothetical case that the  insurer knew the future, the liabilities at time $t$ would be likewise calculated as
\begin{align*}
  L(t) =  \int_{(t,\infty)} \frac{\kappa(t)}{\kappa(s)} \, \d (-C)(s).
\end{align*}
The difference between assets and the hypothetical liabilities determines the value of the insurer's hypothetical own funds,
\begin{align}\label{ALformula}
A(t) - L(t) = \kappa(t) \,( A(0) - L(0)).
\end{align}
The time-dynamic changes of the insurer's  own funds constitute the insurer's profits and losses. In the hypothetical setting  \eqref{ALformula} the profits and losses are solely investment returns earned on the own funds themselves. However, as $A(0)-L(0)$  depends on the future and is nowhere adapted to the currently available information, the insurer actually replaces $A(0)-L(0)$  at each time $t$ with an $\mathcal{F}_t$-measurable proxy $R(t)$.  We call
$$R=(R(t))_{t \geq 0}$$
the \emph{revaluation process}, since it continuously revaluates $A(0)-L(0)$ as the currently available information $\mathcal{F}_t$ increases with time $t$. So in the insurer's balance sheet calculations, the own funds  are actually  of the form
\begin{align}\label{SKappaR}
    S(t) &= \kappa(t) R(t).
\end{align}
We  call $S=(S(t))_{t \geq 0}$ the \emph{surplus process}.
The dynamics of $S$ is driven by two factors,  first the dynamics of $\kappa$, describing investment gains earned on the surplus itself, and second the dynamics of $R$, representing revaluation profits and losses that arise from the continuous revaluation of $A(0)-L(0)$.

The central aim of this paper is to decompose $R$ with respect to different risk sources. Let $\mathbb{D}$ denote the set of adapted  processes on $[0,\infty)$  whose paths are right-continuous and have left limits. We assume that  the stochastic balance sheet model of the insurer rests on a so-called \emph{risk basis}
$$X=(X_1 , \ldots, X_m)\in \mathbb{D}^{d_1 \times \cdots \times d_m} ,$$
  which is a multivariate adapted process composed of so-called \emph{risk factors} $X_1, \ldots, X_m$ such that
  \begin{align}\label{RisXadapted}
    \sigma(R (t)) \subseteq \sigma(X^t), \quad t \geq 0.
  \end{align}
The stopped process $X^t=(X_1^t, \ldots, X_m^t)$ is defined by
\begin{align*}
X_i^t(s) = \mathds{1}_{s \leq t}\, X_i(s) + \mathds{1}_{s > t}\, X_i(t)
\end{align*}
and represents the information that the risk basis $X$ provides at time $t$.
By allowing the risk factors $X_1, \ldots, X_m$ to be multivariate processes with varying dimensions $d_1, \ldots, d_m \in \mathbb{N}$, we have flexibility for grouping of risk factors, where necessary.
Because of \eqref{RisXadapted},  the time-$t$  proxy $R(t)$ of $A(0)-L(0)$ can be seen as a mapping
\begin{align*}
 (t,X^t) \mapsto R(t)
\end{align*}
that assigns to the information $X^t$ at time $t$ the random variable $R(t)$.  Let $L_0$ denote  the set of all real random variables.  This paper generally assumes that there even exists a mapping
$\varrho : \mathcal{D}  \rightarrow L_0$ with $\{X^t : t \geq 0 \} \subseteq \mathcal{D}$ such that
\begin{align}\label{ConstructionR}
 \varrho(X^t)=  R(t), \quad t \geq 0.
\end{align}
By assuming here that the time parameter $t$ itself is not an argument of $\varrho$ but appears only as stopping time parameter in $X^t$, we make sure that the dynamics of $R(t)$ can be completely explained from the increase of information in the $t$-stopped risk basis $X^t$.
\begin{example} In case that the insurer uses risk-neutral valuation in an arbitrage-free financial market, the mapping $\varrho$ takes the form
\begin{align}\label{RhoGleichEQ}
R(t)= \varrho(X^t) = \E^{\mathbb{Q}}  \big[ A(0)-L(0)  \,\big| \, \sigma(X^t) 
\big],
\end{align}
where $\mathbb{Q}$ is a risk-neutral measure, see for example Sheldon \& Smith (2004) for MCEV calculations. In the classical life insurance literature, $\mathbb{Q}$ is rather a conservative valuation measure that represents a technical valuation basis, cf.~Norberg (1999).
\end{example}
\begin{example}
The conditional expectation in \eqref{RhoGleichEQ} is an example of a so-called \emph{conditional risk measure} $\gamma$,
which is a mapping
that assigns to a financial claim $\xi$  a  risk value $\gamma[\xi| \mathcal{G}]$ that is measurable with respect to the current information $\mathcal{G} \subseteq \mathcal{A}$. Popular examples are the Conditional Value at Risk and  the Conditional Average Value at Risk, cf.~F\"{o}llmer \& Schied (2016, section 11).
 Based on such a conditional risk measure $\gamma$, the revaluation process $R$ may be defined by
\begin{align}\label{RhoAsCondRiskMeas}
 R(t) = \varrho(X^t) = \gamma\big[A(0)-L(0)\,\big|\,\sigma(X^t)
 \big].
\end{align}
In the special case where $\gamma$ satisfies the axioms of conditional cash invariance and conditional positive homogeneity, by using \eqref{ALformula}  the right hand side of \eqref{RhoAsCondRiskMeas} can be transformed to
\begin{align*}
   \frac{1}{\kappa(t)}A(t)- \frac{1}{\kappa(t)}\gamma\big[  L(t)  \,\big| \,\sigma(X^t) 
   \big] ,
\end{align*}
provided that $\kappa$ is adapted to the natural filtration of $X$.
The latter equation means that it is only the liabilities that make it necessary to use the proxy $R(t)$ in the insurer's balance sheet.
\end{example}
\begin{example}
In classical with-profit life insurance, there typically exists a deterministic operator $F$ such that
\begin{align}\label{RhoExp3}
  R(t)(\omega)= \varrho(X^t)(\omega) =   F(X^t(\omega)), \quad \omega \in \Omega,
\end{align}
provided  that we define $X$ suitably,
 see Jetses \& Christiansen (2021). 
\end{example}
The central aim of this paper is to decompose $R(t) =\varrho(X^t)$ as
\begin{align}\label{GeneralAdditDecomp}
  R(t) = R(0) + D_1(t) + \cdots + D_m(t), \quad t \geq 0,
\end{align}
where $D_1, \ldots, D_m$ are adapted processes with $D_i(0)=0$ and such that $D_i(t)$ describes the contribution of  risk factor $X_i$ to the dynamics of $R$. The first addend $R(0)$ represents the initial own funds, which are not decomposed here.
Equation \eqref{GeneralAdditDecomp} is equivalent to the additive decomposition
\begin{align}\label{totalSurplus2}
  S(t) = \kappa(t) \, S(0) +  \kappa(t)\, D_1(t)+ \cdots + \kappa(t)\, D_m(t), \quad t \geq 0,
\end{align}
for the surplus process $S$. The first addend $\kappa(t) \, S(0) $ represents the time-$t$ value of the initial own funds $S(0)= R(0)$. The remaining addends $\kappa(t)\, D_1(t),  \ldots ,\kappa(t)\, D_m(t)$ describe the time-$t$ values of the contributions that the risk factors $X_1, \ldots, X_m$ make to the dynamics of $S$.  
\section{An axiomatic perspective on profit and loss decompositions}
 In this section we postulate three desirable properties for the decomposition $D=(D_1,\ldots, D_m)$ of $R=(\varrho(X^t))_{t \geq 0}$. Our first axiom is \emph{additivity}, which was already postulated in    \eqref{GeneralAdditDecomp}. Additive decompositions are intuitively interpretable and make it possible to redistribute profits and losses among different parties.
  Our second axiom is  \emph{normalization}. Recall that our construction \eqref{ConstructionR} of the revaluation process follows the principle that revaluation profits and losses are a result of information growth. If the risk basis $X$ were constant on some interval $(a,b]$, then the  information growth $t \mapsto \sigma(X^t)$ is de facto halted on $(a,b]$ and the revaluation process $t \mapsto \varrho(X^t)$ produces no new profits and losses.  By postulating the latter principle also for the partial profit and loss contributions of the risk factors, we arrive at the normalization axiom:
 \begin{align*}
               X_i \textrm{ is constant on } (a,b]  \quad &\Longrightarrow \quad D_i \textrm{ is constant on } (a,b]
 \end{align*}
 for $i \in \{1, \ldots, m\}$ and $(a,b] \subset [0,\infty)$.
Our third axiom is \emph{stability} of $D$ with respect to perturbations in the empirical observation of the risk basis $X$. The potential deformations of $X$ shall be given by the set
$$\mathcal{X} \subseteq  \mathbb{D}^{d_1 \times \cdots \times d_m}.$$
The perturbed empirical observations $Y \in \mathcal{X}$ replace the original risk basis $X$ and the insurer calculates the revaluation process de facto as $t \mapsto \varrho( Y^t)$.
In order that the latter process is well-defined, we need to  assume here that the domain $\mathcal{D}$ of  $\varrho$ includes the set
\begin{align*}
 \{ Y^t: Y \in \mathcal{X}, t \geq 0 \}  \subseteq \mathcal{D}.
\end{align*}
The perturbed empirical observations $Y \in \mathcal{X}$ replace the original risk basis $X$
also  in the  calculation of the decomposition $D$. Therefore, we need to define profit and loss decompositions not only for $X$ but for all $Y \in \mathcal{X}$. Let $\mathbb{D}_0$ be the subset of those processes of $\mathbb{D}$ that are zero at time zero. For the sake of a simple notation we assume from now on that $\mathcal{X}$ includes also the original risk basis $X$.
\begin{definition}\label{DefDecScheme}
We call $\delta:   \mathcal{X} \rightarrow \mathbb{D}^m_0$ a \emph{decomposition scheme for $\varrho$} if $\delta$ assigns to each $Y \in \mathcal{X} $ a decomposition $\delta(Y)=(\delta_1(Y), \ldots, \delta_m(Y)) $ of $(\varrho(Y^t))_{t \geq 0}$.
\end{definition}
Recall that the decomposition $\delta(X)$ is inaccessible to the insurer and de facto replaced by a perturbed decomposition $\delta (Y)$. So far, our axioms of additivity and normalization refer to the decomposition $\delta(X)$ only. It is desirable that $\delta (Y)$ is also additive and normalized.
Additivity makes the decomposition easily interpretable and allows the insurer to  redistribute the profits and losses  among different parties. Normalization implements the intuitive principle that there are no revaluation profits and losses without information growth. Therefore, we extend the additivity and normalization axioms to the whole set $\mathcal{X}$.
\begin{enumerate}
  \item[(\textbf{A})] \textbf{Additivity}: 
     For all $Y \in \mathcal{X}$  and $t \geq 0$ let
        $$ \varrho (Y^t) - \varrho (Y^0) =  \delta_i(Y)(t) + \cdots +\delta_m(Y)(t) $$
        almost surely.
\end{enumerate}
\begin{enumerate}
  \item[(\textbf{N})] \textbf{Normalization}: 
      For all $Y \in \mathcal{X}$ and $(a,b] \subset [0,\infty)$
      let the implication
            \begin{align*}
               Y_i \textrm{ is constant on } (a,b]  \quad &\Longrightarrow \quad \delta_i(Y) \textrm{ is constant on } (a,b]
            \end{align*}
            be satisfied for all $i \in \{1, \ldots, m\}$.
\end{enumerate}
Before we can formalize our stability axiom, we need to further specify the potential perturbations $\mathcal{X}$ of the risk basis $X$. 
\begin{enumerate}
\item[(1)] Data collection and data processing efforts lead to time lags in the observation of $X$.  These time lags may be even asynchronous in the different components of $X$.
\item[(2)] Empirical time series data of $X$ may be  available as time quantized data only.
\end{enumerate}
Asynchronous time lags and time quantization  in the observation of $X$ do not just mean that the current information $X^t$ is behind schedule, but they can fundamentally change the structure of  the information: Asynchronous time lags may permutate the order of events, and  time quantization means that the observed information on $X$ is generally incomplete.
We model such perturbations of the risk basis by non-decreasing and right-continuous functions $\tau_i: [0,\infty) \rightarrow [0,\infty) $, $i \in \{1, \ldots, m\}$ with  $\tau_i(t) \leq t$ for all $ t \geq 0$.
The multivariate function
\begin{align*}
  \tau(t) = (\tau_0(t), \ldots , \tau_m(t))
\end{align*}
 is called a \emph{delay}. The difference $t- \tau_i(t)$ represents the current time lag at time $t$ in component $i$. The jumps of  $\tau_i$ represent time quantization in the observation of process $X_i$. For example, $\tau_i(t) = \lfloor t \rfloor$ means that $X_i \circ \tau_i$ reveals us only the values of $X_i$ at integer times.
In case that the empirical observation of the risk basis $X$ suffers from the delay $\tau$, the insurer actually observes
\begin{align*}
X \diamond \tau := (X_0 \circ \tau_0 , \ldots,  X_m \circ \tau_m)
\end{align*}
instead of $X$.
The image $\tau_i([0,t])$ gives the time points for which we come to know the values of $X_i$.
If the image $\tau_i([0,t])$ is a dense subset of $[0,t]$ and includes time $t$, then we already know the whole process $X_i^t$ because right-continuous process are uniquely defined by their values on dense subsets.

\begin{proposition}\label{tauIncresToInf}
Let $(\tau^n)_{n \in \mathbb{N}}$ be a  sequence  of 
delays that satisfy
\begin{align}\label{DefIncMonToId}
\tau^n_i([0,t]) \subseteq \tau^{n+1}_i([0,t]),\; n \in \mathbb{N} \quad \textrm{ and }  \quad \overline{\bigcup_{n \in \mathbb{N}} \tau^n_i([0,t])} =  [0,t]
\end{align}
for each $t \geq 0$ and  $i \in \{1, \ldots, m\}$.
 Then for $t>0$ and $n \rightarrow \infty$ it holds that
  \begin{enumerate}
    \item[(a)]  $ \sup_{0 \leq s \leq t} | s -\tau_i^n(s) | \rightarrow 0$ for all   $ i \in \{1, \ldots, m\}$,
    \item[(b)]  $ X \diamond \tau^n(t-) \rightarrow X(t-) $ for all   $ i \in \{1, \ldots, m\}$,
    \item[(c)] the sequence of sigma-algebras $\sigma((X\diamond \tau^n)^{t-})$, $n \in \mathbb{N}$, is non-decreasing and converges to $ \bigvee_{n \in \mathbb{N}} \sigma((X\diamond \tau^n)^{t-})  = \sigma(X^{t-})$.
\end{enumerate}
\end{proposition}
\begin{proof}
At first, we show the pointwise convergence $\tau_i^n(t) \rightarrow  t$ for each time $t$. Note that $\tau^n_i(0)=0$ by definition. Suppose that there is a time point $t_0>0$ where $\liminf_{n \rightarrow \infty}  \tau_i^n(t_0)= s_0 < t_0$. Then the monotony assumption in \eqref{DefIncMonToId} implies that $\tau_i^n([0,t_0]) \subseteq [0,s_0]$ for all $n \in \mathbb{N}$, which is a contradiction to  the second assumption in \eqref{DefIncMonToId}. Thus,  we necessarily have $\liminf_{n \rightarrow \infty}  \tau_i^n(t)=  t$ for each time $t \geq 0$, which implies  $\tau_i^n(t) \rightarrow  t$ since  $\tau^n_i(t) \leq t$ by definition.

Let $\varepsilon >0$ be arbitrary but fixed. The fact that delays have the identity function as upper bound implies that $\tau_i^n(s +\varepsilon/4) \leq s + \varepsilon/4$ for all $s \geq 0$. On the other hand, because of the pointwise convergence of $\tau_i^n$ to the identity function (as we showed before),  there exists for each time $s_0$ an integer $n_0 \in \mathbb{N}$ such that $\tau_i^n(s_0 -\varepsilon/4) \geq s_0 - \varepsilon/2$ for all $n \geq n_0$. That means that
$$|s- \tau_i^n(s)| \leq | s_0 - \tau_i^n(s_0- \varepsilon/4) | +| s_0 - \tau_i^n(s_0+ \varepsilon/4) |  \leq \varepsilon$$
for all $|s-s_0| < \varepsilon/4$ and $n \geq n_0$ since $\tau^n_i$ is a nondecreasing function. Hence, on each compact interval $[0,t]$ there is an integer $n_1 \in \mathbb{N}$ such that $\sup_{s \in [0,t]} | s -\tau_i^n(s) | \leq \varepsilon$ for all $n \geq n_1$.  This proves assertion (a).

Statement (b) follows from (a) and the left-continuity of $s \mapsto X(s-)$.

The monotony of  the sequence $\sigma((X\diamond \tau^n)^{t-})$, $n \in \mathbb{N}$, is a direct consequence of the  monotony assumption in \eqref{DefIncMonToId}.  Moreover,
since delays have the identity function as upper bound, we have $\sigma((X\diamond \tau^n)^{t-})  \subseteq \sigma(X^{t-})$ for all  $n \in \mathbb{N}$.  As $X$ is right-continuous, it is uniquely defined by its values on any dense subset of the time line.
 So assumption \eqref{DefIncMonToId} implies that
\begin{align*}
 \sigma(X_i^{t-}) = \sigma\bigg( X_i(u) : u \in  \bigcup_{n \in \mathbb{N}} \tau^n_i([0,t))  \bigg) \subseteq \bigvee_{n \in \mathbb{N}}  \sigma((X_i\circ \tau^n_i)^{t-}).
\end{align*}
This proves assertion (c).
\end{proof}
Suppose that $(X\diamond \tau^n)^t$ describes the actually observable information about $X^t$ at time $t$. The first condition in \eqref{DefIncMonToId} implies that  the observable information is monotonously increasing in $n$, and the second condition in \eqref{DefIncMonToId} implies that the observable information is increasing to full information as $n \rightarrow \infty$, see Proposition \ref{tauIncresToInf}(c).
If  $(\tau^n)_{n \in \mathbb{N}}$ is a sequence of 
 delays that  satisfies \eqref{DefIncMonToId}, then we say that   \emph{$(\tau^n)_{n \in \mathbb{N}}$ is a refining sequence of 
 delays that increase to identity}.

In the following let $\mathbb{T}$ denote the set of all delays, and let
$\mathcal{T} \subseteq  \mathbb{T}$ be a non-empty subset that describes the potential perturbations that shall be accounted for in particular applications. So we have
$$ \mathcal{X} = \{ X \diamond \tau: \tau \in \mathcal{T}\}  \cup \{X\}.$$
A delay $\tau \in \mathcal{T}$ in the empirical observation of $X$ has the consequence that the revaluation process $t \mapsto \varrho(X^t)$ and its decomposition $\delta(X)$ are de facto replaced by $t \mapsto \varrho((X \diamond \tau)^t)$ and $\delta(X \diamond \tau)$. In insurance practice it is desirable that $\delta (X \diamond \tau)$ and $\delta(X)$  are not too different.
\begin{enumerate}
  \item[(\textbf{S})] \textbf{Stability at $X$}:
     For  any refining sequence      $(\tau^n)_{n \in \mathbb{N}} \subseteq \mathcal{T}$   of  delays  that increase to identity,  let
      \begin{align}\label{ContAxiom}
      \delta( X \diamond \tau^n)(t-)  \stackrel{\mathbb{P}}{\longrightarrow } \delta( X )(t-).
      \end{align}
\end{enumerate}
Condition \eqref{ContAxiom} uses  left limits because we generally have  $ X\diamond \tau^n(t-)  \rightarrow X(t-)$ according to Proposition \ref{tauIncresToInf}(b), whereas $ X\diamond \tau^n(t)$ is not necessarily converging to $X(t)$. Depending on the choice of $\mathcal{T}$, axiom (S) can describe different levels of stability. For $\mathcal{T}= \mathbb{T}$ we obtain the strongest form of stability. For $\mathcal{T}$ strictly smaller than $\mathbb{T}$ we get weaker forms of stability. For example, if we want to focus on continuous delays only, then we would set $\mathcal{T}= \{ \tau \in \mathbb{T}: \tau \textrm{ is continuous}\}$.

\section{Decompositions based on sequential updating}

For the moment, we put the axiomatic perspective of the previous section aside and turn to
the ISU decomposition principle introduced by Jetses \& Christiansen (2021). We largely follow the definition of Jetses \& Christiansen (2021) but expand the perspective from just a fixed time point to the full time line.

Recall that the $t$-stopped process $X^{t}=(X^t_1, \ldots, X_m^t)$  represents the currently available information on the risk factors $X_1, \ldots, X_m$  at time $t$.
Suppose that the  information updates on the risk factors $X_1, \ldots, X_m$ are asynchronously time-lagged  with $t_1, \ldots , t_m  \leq t$ being the current update statuses of each risk factor at time $t$.
Then $(X^{t_1}_1, \ldots, X^{t_m}_m)$ represents the currently observable information and
\begin{align*}
U(t_1, \ldots, t_m):=    \varrho((X^{t_1}_1, \ldots, X^{t_m}_m) )
\end{align*}
is the de facto value of the revaluation process at time $t$. %
We need to assume here that the domain $\mathcal{D}$ of $\varrho$ includes the set
$$ \{ (X^{t_1}_1, \ldots, X^{t_m}_m) : t_1, \ldots , t_m  \geq 0\} \subseteq \mathcal{D}.$$
 We denote $U=( U(t_1, \ldots, t_m))_{ t_1, \ldots , t_m  \geq 0}$
as the \emph{revaluation surface}.
We can recover the revaluation process $R$ from the revaluation surface $U$ by
\begin{align*}
  R(t) = U(t, \ldots, t), \quad  t\geq 0.
\end{align*}
For any  unbounded partition $\pi = \{0=s_0 < s_1 < \cdots  \}$ of the interval $[0,\infty)$ we
  can build the telescoping series
\begin{align*}
 R(t)-R(0) &= U(t, \ldots, t) -U(0, \ldots,0)\\
  &= \sum_{l=0}^{\infty} \Big( U(s_{l+1}\wedge t,s_{l}\wedge t, \ldots,s_{l}\wedge t)-U(s_{l}\wedge t,s_{l}\wedge t, \ldots,s_{l}\wedge t)\Big) \\
   &\quad +  \sum_{l=0}^{\infty} \Big( U(s_{l+1}\wedge t,s_{l+1}\wedge t, s_{l}\wedge t,\ldots,s_{l}\wedge t)-U(s_{l+1}\wedge t,s_{l}\wedge t, \ldots,s_{l}\wedge t)\Big) \\
  &\quad + \cdots \\
  & \quad + \sum_{l=0}^{\infty} \Big( U(s_{l+1}\wedge t,\ldots, s_{l+1}\wedge t, s_{l+1}\wedge t)-U(s_{l+1}\wedge t,\ldots, s_{l+1}\wedge t, s_{l}\wedge t)\Big)
 \end{align*}
 for each $t \geq 0 $.  These sums always exist since they have at most finitely many non-zero addends. It is natural here to interpret the $m$ different sums as an additive  decomposition $R(t)-R(0)=D_1(t)+  \cdots + D_m(t)$,   since the $i$-th sum collects exactly the  information updates for the $i$-th risk factor.
\begin{definition}\label{DefSUDecomp}
The  $m$-dimensional process  $D=(D_1, \ldots, D_m)$  defined by
\begin{align}\label{SUdecomposition}\begin{split}
  D_1(t) &= \sum_{l=0}^{\infty} \Big( U(s_{l+1}\wedge t,s_{l}\wedge t, \ldots,s_{l}\wedge t)-U(s_{l}\wedge t,s_{l}\wedge t, \ldots,s_{l}\wedge t)\Big), \\
  & \cdots  \\
    D_{m}(t) &= \sum_{l=0}^{\infty}  \Big( U(s_{l+1}\wedge t,\ldots, s_{l+1}\wedge t, s_{l+1}\wedge t)-U(s_{l+1}\wedge t,\ldots, s_{l+1}\wedge t, s_{l}\wedge t)\Big)
 \end{split}\end{align}
  is called  the \emph{SU (sequential updating) decomposition}  of $(\varrho(X^t))_{t \geq 0}$ with respect to $\pi$.
\end{definition}
The SU decomposition principle is frequently used in the economics literature, cf.~Fortin et al.~(2011), Shorrocks (2013) and Biewen (2014).
Insurers regularly apply the SU decomposition principle for the change analysis of profits and losses, cf.~Bashford \& Dubischar (2020). Candland \& Lotz (2014) denote the SU decomposition as the `waterfall decomposition', which reflects the fact that the decomposition results are often plotted in  waterfall charts, illustrated in Figure \ref{waterfall}.
\begin{figure}[h]
\begin{center}
\includegraphics[scale=0.4]{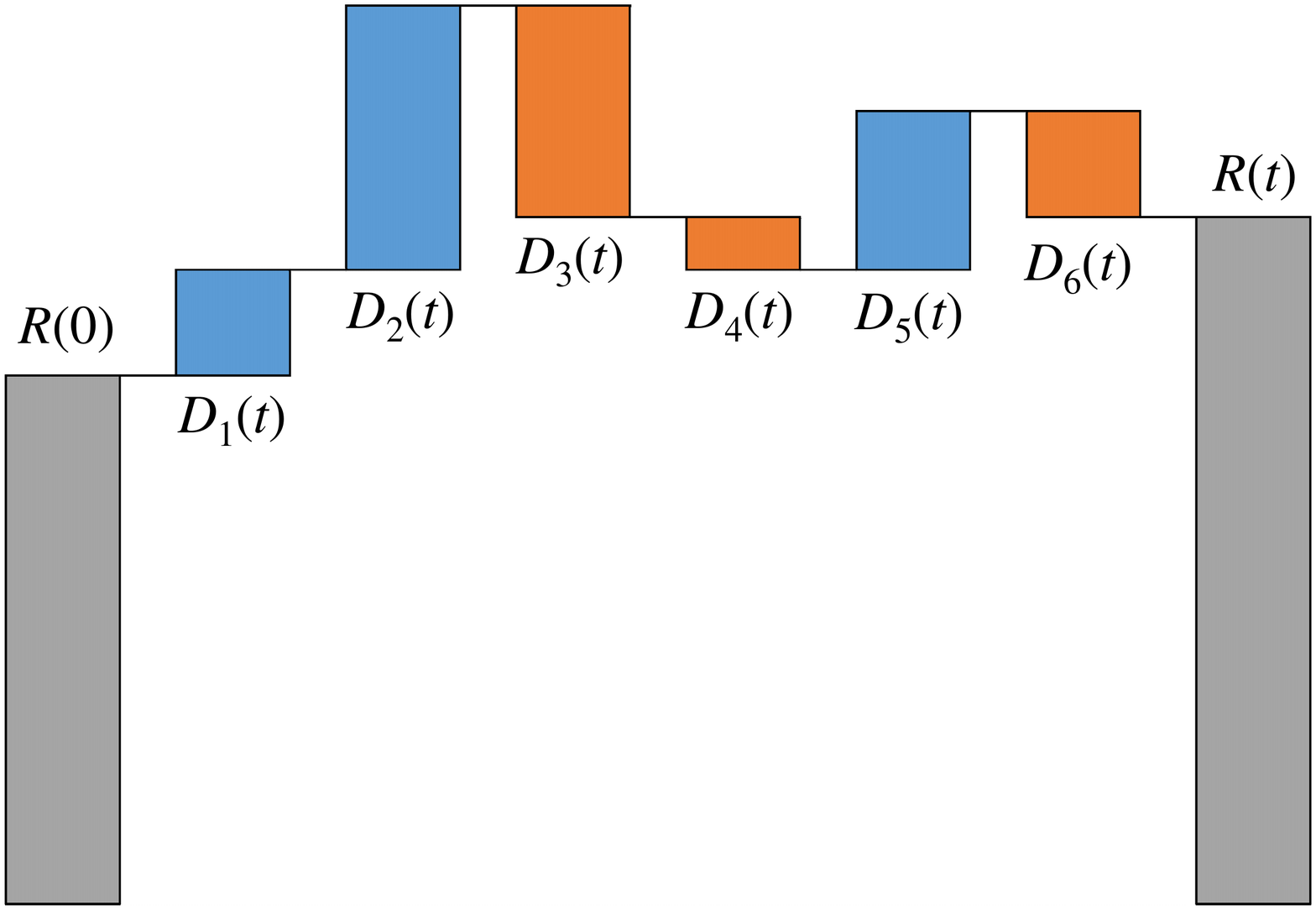}
\caption{Example of a waterfall chart showing decomposed profits and losses}\label{waterfall}%
\end{center}
\end{figure}
\begin{proposition}\label{LemmaSUisAN}
 For any unbounded partition $\pi$, the corresponding SU decomposition of $(\varrho(X^t))_{t \geq 0}$ is additive and normalized.
\end{proposition}
\begin{proof}
The additivity follows from the construction of the SU decomposition as a telescoping sum. If $X_1$ is constant on some interval $(a,b]$, then $U(s_{l+1}\wedge t,s_{l}\wedge t, \ldots,s_{l}\wedge t)$ and $U(s_{l}\wedge t,s_{l}\wedge t, \ldots,s_{l}\wedge t)$ are also constant on $(a,b]$, so $D_1$ is constant on $(a,b]$. For $D_2, \ldots, D_m$ we can argue likewise. Thus, the SU decomposition is normalized.
\end{proof}

In  \eqref{SUdecomposition}   the information on $ X $ is updated in each time step in a specific order, starting with risk factor $X_1$, then updating $X_2$, and so on. So the SU decomposition depends on the update order, which is its main disadvantage. Jetses \& Christiansen (2021) suggest to make the impact of the update order vanish by increasing the number of updating steps to infinity.
 Let  $\pi^n=\{ 0=s^n_0 < s^n_1 < \cdots  \}$, $n \in \mathbb{N}$, be an increasing (i.e.~$\pi^n\subset \pi^{n+1}$ for all $n$) sequence of unbounded (i.e.~$\lim_{ k\rightarrow \infty} s^n_{k} =\infty$)
 partitions of $[0,\infty)$  with vanishing step lengths (i.e.~$  \lim_{ n\rightarrow \infty} \max_{k } |s^n_{k+1}-s^n_{k}| =0$).
 For each $n \in \mathbb{N}$ let  $D^n=(D_1^{n}, \ldots ,D_{m}^{n})$ be the SU decomposition of $R$ with respect to $\pi^n$. We are looking for a process $D$ that satisfies
\begin{align}\label{ISUdecomposition}\begin{split}
 D_i(t) = \plim_{n \rightarrow \infty} D^{n}_i(t), \quad i \in \{1, \ldots, m\},,\, t \geq 0.
\end{split}\end{align}
\begin{definition}\label{DefISUdecomp}
 Let $\Pi=(\pi^n)_{n \in \mathbb{N}}$ be an increasing sequence of unbounded partitions of $[0,\infty)$ with vanishing step lengths.
 If  $D \in \mathbb{D}^m_0 $ is a process that satisfies \eqref{ISUdecomposition}, then we call $D$ the \emph{ISU (infinitesimal sequential updating) decomposition}  of $(\varrho(X^t))_{t \geq 0}$ with respect to $\Pi$.
\end{definition}
Since we  assume that $D \in \mathbb{D}^m_0 $, the ISU decomposition is unique up to indistinguishability. However, we do not have a guarantee that  it really  exists.

\begin{proposition}\label{LemmaISUisAN} For any increasing sequence of unbounded partitions $\Pi$, the corresponding ISU decomposition of $(\varrho(X^t))_{t \geq 0}$  is  additive and normalized, provided that it exists.
\end{proposition}
\begin{proof}
As the SU decomposition is additive, see Proposition \ref{LemmaSUisAN}, equation  \eqref{ISUdecomposition} implies that
$$ \sum_i D_i(t) = \sum_i \plim_{n \rightarrow \infty} D^{n}_i(t) =\plim_{n \rightarrow \infty} \sum_i D^{n}_i(t)  = \varrho(X^t)-\varrho(X^0).$$
If $X_i$ is constant on $(a,b]$, then we get
$$  D_i(t) -D_i(s)  = \plim_{n \rightarrow \infty} (D^{n}_i(t) - D^{n}_i(s)) =0$$
for $a\leq s \leq t \leq b$ since the SU decomposition is  normalized, cf.~Proposition \ref{LemmaSUisAN}.
\end{proof}

\begin{definition}[ISU decomposition scheme]
The ISU decomposition scheme $\delta^{ISU}$ with respect to $\Pi$ is the mapping that assigns to each
argument $Y \in \mathcal{X}$ the ISU decomposition of $(\varrho(Y^t))_{t \geq 0}$ with respect to $\Pi$.
\end{definition}
We say that $\delta^{ISU}$ exists if the ISU decompositions of $(\varrho(Y^t))_{t \geq 0}$  exist for all $Y \in \mathcal{X}$, which in particular implies that
the domain $\mathcal{D}$ of $\varrho$ has to include the set
$$ \{ (Y^{t_1}_1, \ldots, Y^{t_m}_m) : t_1, \ldots , t_m  \geq 0,  Y \in \mathcal{X}\} \subseteq \mathcal{D}.$$
We say that  a delay $\tau$ is \emph{phased}  if there exists an unbounded partition  $\{ 0=s_0 < s_1 < \cdots  \}$ of $[0,\infty)$ such that on each interval $(s_l,s_{l+1}]$ of the partition at most one component of $\tau$ is non-constant. As we show in the following proposition, if ISU decompositions are stable with respect to phased delays, then they are generally invariant with respect to the choice of $\Pi$ and invariant with respect to the choice of the update order.
\begin{proposition}[Invariance]\label{CorollaryPiArbitrary}
Suppose that  $\mathcal{T}$  includes a refining sequence   $(\tau^n)_{n \in \mathbb{N}}$   of  phased delays  that increase to identity.
Let $\delta^{ISU}$ and  $\widetilde{\delta}^{ISU}$ denote the ISU decomposition schemes with respect to $\Pi$ and $\widetilde{\Pi}$, assuming  that they exist. The update order in $\widetilde{\delta}^{ISU}$ may be permutated. If $\delta^{ISU}$ and  $\widetilde{\delta}^{ISU}$ are stable at $X$, then
\begin{align*}
  \delta^{ISU} (X) = \widetilde{\delta}^{ISU} (X)
\end{align*}
almost surely.
\end{proposition}
\begin{proof}
For each $n \in \mathbb{N}$,
let $(a^n_{l,i},b^n_{l,i}]$, $l \in \mathbb{N}_0$, $i \in \{1, \ldots, m\}$ be a partition of $[0,\infty)$ such that $(\tau^n_j)_{j \neq i}$ is constant on $(a^n_{l,i},b^n_{l,i}]$, $l \in \mathbb{N}_0$.
Since $\delta^{ISU}$ is normalized, see Proposition \ref{LemmaISUisAN}, we have
\begin{align*}
  \delta^{ISU}_i(X\diamond \tau^n)(t) &=  \sum_{l} \big( \varrho( (X \diamond \tau^n) ^{b^n_{l,i}\wedge t})- \varrho((X \diamond \tau^n)^{a^n_{l,i}\wedge t})\big)
\end{align*}
for all $i \in \{1, \ldots, m\}$ and $t \geq 0$.
Likewise we also have
\begin{align*}
  \widetilde{\delta}^{ISU}_i(X\diamond \tau^n)(t) &=  \sum_{l} \big( \varrho( (X \diamond \tau^n) ^{b^n_{l,i}\wedge t})- \varrho((X \diamond \tau^n)^{a^n_{l,i}\wedge t})\big)
\end{align*}
for all $i \in \{1, \ldots, m\}$ and $t \geq 0$, so
\begin{align*}
   \delta^{ISU}_i(X\diamond \tau^n)(t) =  \widetilde{\delta}^{ISU}_i(X\diamond \tau^n)(t)
\end{align*}
for all $i \in \{1, \ldots, m\}$ and  $t \geq 0$. Since both of the latter processes are right-continuous by definition, we even have  $   \delta^{ISU}_i(X\diamond \tau^n)=  \widetilde{\delta}^{ISU}_i(X\diamond \tau^n)$ almost surely.
By using the stability of  $\delta^{ISU}$ and  $\widetilde{\delta}^{ISU}$ at $X$, we obtain
\begin{align*}
  \delta^{ISU}_i(X)(t-) = \plim_{n \rightarrow \infty}  \delta^{ISU}_i(X\diamond\tau^n )(t-)  = \plim_{n \rightarrow \infty}  \widetilde{\delta}^{ISU}_i(X\diamond\tau^n )(t-) =  \widetilde{\delta}^{ISU}_i(X)(t-)
\end{align*}
almost surely for each $t>0$. Since $\delta^{ISU}_i(X)$ and $\widetilde{\delta}^{ISU}_i(X)$  are elements of $\mathbb{D}_0$, we can conclude that  $   \delta^{ISU}_i(X)=  \widetilde{\delta}^{ISU}_i(X)$ almost surely for  all $i \in \{1, \ldots, m\}$.
\end{proof}

\section{Characterization of  additive, normalized, stable decompositions}

This section characterize decompositions that satisfy the axioms (A), (N), and (S).
\begin{theorem}[Uniqueness]\label{Prop:UniqueANS}
Let  the set $\mathcal{T}$ include  a refining sequence   $(\tau^n)_{n \in \mathbb{N}}$   of  phased delays  that increase to identity.
Suppose that $\delta$ and $\widetilde{\delta}$ are additive, normalized decomposition schemes that are stable at $X$.
 Then  we almost surely have
\begin{align*}
  \delta (X) = \widetilde{\delta} (X).
\end{align*}
\end{theorem}
\begin{proof} The proof is analogous to the proof of Proposition \ref{CorollaryPiArbitrary}, since the latter proof uses only the additivity and normalization of ISU decompositions but not their specific structure.
\end{proof}
\begin{theorem}[ISU representation]\label{Prop:EquivISUandANS_F}
Let  the set $\mathcal{T}$ include  a refining sequence   $(\tau^n)_{n \in \mathbb{N}}$   of  phased delays  that increase to identity.
Suppose that $\delta$ is any additive,  normalized decomposition scheme that is stable at $X$. If $\delta^{ISU}$ exists for some $\Pi$ and is stable at $X$, then we almost surely have
\begin{align*}
  \delta (X) = \delta^{ISU} (X).
\end{align*}
\end{theorem}
\begin{proof}
The theorem follows from Proposition \ref{LemmaISUisAN} and Theorem \ref{Prop:UniqueANS}.
\end{proof}

Theorem \ref{Prop:EquivISUandANS_F}  shows us that the axioms (A), (N), (S)  necessarily lead to ISU decompositions if the latter are sufficiently stable.  For given mapping $\varrho$ and  risk basis $X$, the above results suggest to determine the decomposition $D$ for $(\varrho(X^t))_{t \geq 0}$  as follows:
\begin{enumerate}
\item[(1)] Define the set $\mathcal{T}$ of delays  that shall be considered.
\item[(2)] Calculate $\delta^{ISU}$ on $\mathcal{X}$ for some choice of $\Pi$.
\item[(3)] Verify that $\delta^{ISU}$ is stable at $X$.
\item[(4)] Set $D= \delta^{ISU}(X)$.
\end{enumerate}
If all steps can be successfully completed, then $D$ is the only decomposition of $(\varrho(X^t))_{t \geq 0}$ that is additive, normalized, and stable.
In insurance practice it is usually  not really necessary to numerically calculate the ISU decomposition $D= \delta^{ISU}(X)$, since it can be approximated by SU decompositions, see \eqref{ISUdecomposition}.
SU decompositions are easy to calculate; the insurer just needs to sufficiently increase the frequency of the usual balance sheet calculations so that the addends in \eqref{SUdecomposition} become available.
 Yet, before such an approximation is used, one should verify on a theoretical level that the ISU decomposition scheme $\delta^{ISU}$ indeed exists and is  stable at $X$.

\section{Examples}\label{SectionExamples}

This section looks at three examples that illustrate the theoretical results of the previous sections. In all three examples we aim to split the profits and losses of a life insurance portfolio into investment surplus and mortality surplus.
For improving the readability, the mathematical proofs are postponed to  the next section.

We consider  a portfolio of 
endowment insurances that start at time zero and run for $T$ years. Let $(p_j)_{j}$ be the individual lump sum premiums at the start of the insurance contracts,  and let $(b_j)_{j}$ be the endowment benefits in case of survival till time $T$. The status of the $j$-th insured shall be given by a right-continuous jump process $N_j$ that starts from zero at time  zero and jumps to $1$ at the time of death. The so-called mortality intensity for the $j$-th insured is defined as
\begin{align*}
  \lambda_j(t) = \lim_{h \downarrow 0} \frac{\mathbb{P}(N_j(t+h)-N_j(t)=1)}{\mathbb{P}(N_j(t)=0)}, \quad t \geq 0,
\end{align*}
assuming that the limit exists.  For the sake of a convenient notation we set  $I_j=1-N_j$.
At time zero, the hypothetical own funds of the insurance portfolio equal
\begin{align*}
  A(0) -L(0) =  \sum_{j} p_j -  \sum_{j} \frac{I_j(T)}{\kappa(T)} b_j\,.
\end{align*}

\subsection*{Risk-neutral valuation by a martingale measure}

With the aim  to  decompose the revaluation process $R$ into an  investment part and a mortality part, we set
 $$X=(X_1,X_2)=(\Phi, N),$$ where  $N=(N_j)_j$ is the vector of all individual counting processes.  Suppose that $R$ is given by
\begin{align*}
  R(t)=  \varrho(X^t) = \E^{\mathbb{Q}} \big[A(0)-L(0) \big| \sigma(X^t) 
  \big]
\end{align*}
for some risk-neutral martingale measure $\mathbb{Q}$, see example \eqref{RhoGleichEQ}. 
By stopping the processes $X_1=\Phi$ and $X_2=N$ at asynchronous time points $t_1$ and $t_2$, we obtain the revaluation surface
\begin{align*}
  U(t_1,t_2)=  \varrho((X_1^{t_1},X_2^{t_2})) = \E^{\mathbb{Q}} \big[A(0)-L(0) \big| \sigma(X_1^{t_1},X_2^{t_2}) 
  \big].
\end{align*}
Let $\Phi$   be a Brownian motion under $\mathbb{Q}$ with the representation
 $$\d\Phi(t) = \mu \,\d t + \sigma \, \d W^{\mathbb{Q}}(t) $$
so that $\kappa$ corresponds to a Black-Scholes model.  We assume that $\Phi, N_1, N_2, \ldots$  are stochastically independent under $\mathbb{Q}$. Let $(\lambda^{\mathbb{Q}}_j)_j$ be the transition intensities of $(N_j)_j$ under $\mathbb{Q}$, so that $\d N_j(t)$ has the $\mathbb{Q}$-compensator
\begin{align*}
  \d C^{\mathbb{Q}}_j(t) = I_j(s-) \lambda^{\mathbb{Q}}_j(s) \,\d s
\end{align*}
with respect to the natural filtration of $X$. 
 We focus on continuous delays only and set
 $$ \mathcal{T}= \{\tau \in \mathbb{T}: \tau \textrm{ is continuous} \}. $$
\begin{proposition}\label{LemmaExp1}
The ISU decomposition scheme $\delta^{ISU}$ exists and is stable at $X$. It  holds that
\begin{align}\label{D1Exp1}
  \delta^{ISU}_1(X)(t) & = \sum_{j} b_j  \int_{(0,t\wedge T]}  \frac{I_j(s)  \, \mathrm{e}^{-\int_s^T \lambda^{\mathbb{Q}}_j(u)\,\d u} }{\kappa(s)\,\mathrm{e}^{  (T-s)(\mu-\sigma^2) }}   \, \sigma\,\d W^{\mathbb{Q}}(s),\\
  \label{D2Exp1}
  \delta^{ISU}_2(X)(t) & =  \sum_{j} b_j \int_{(0,t\wedge T]} \frac{ \mathrm{e}^{-\int_s^T \lambda^{\mathbb{Q}}_j(u)\,\d u} }{\kappa(s)\,\mathrm{e}^{  (T-s)(\mu-\sigma^2) }} \, \d (N_j-C^{\mathbb{Q}}_j)(s)
\end{align}
almost surely for all $t \geq 0$.
\end{proposition}
According to  Proposition \ref{LemmaExp1} and Theorem \ref{Prop:EquivISUandANS_F}, the formulas \eqref{D1Exp1}\&\eqref{D2Exp1}  represent
the only decomposition of $(\varrho(X^t))_{t \geq 0}$ that is additive, normalized, and stable.
The decomposition \eqref{D1Exp1}\&\eqref{D2Exp1} actually coincides with the martingale representation of $R$, so it equals the so-called MRT decomposition of Schilling et al.~(2020). While the `meaningful risk decomposition' axioms of Schilling et al.~(2020) are not sufficient to uniquely identify the MRT decomposition, our axioms of additivity, normalization and stability  uniquely characterize the MRT decomposition here.

\subsection*{Conservative valuation by a conditional risk measure}

We redefine $R$ as conditional $\mathbb{P}$-expectation  plus  a safety margin according to the standard deviation principle,
\begin{align*}
  R(t) =\varrho(X^t) =\E[A(0)-L(0)| \sigma(X^t) 
  ] + \alpha\,\sqrt{ \mathrm{Var}[A(0)-L(0)| \sigma(X^t) 
  ]},
\end{align*}
where $\alpha >0$ is a fixed parameter that controls the safety level. This formula is an example of a conditional risk measure  according to \eqref{RhoAsCondRiskMeas}.
We still use
 $$X=(\Phi, N)$$  as our risk basis.
Under the real-world measure $\mathbb{P}$ the Brownian motion $\Phi$ may  have a different drift,
 $$\d\Phi(t) = r \,\d t + \sigma \, \d W(t). $$
 We assume that $\Phi, N_1, N_2, \ldots $ are stochastically independent under $\mathbb{P}$.
The process $\d N_j(t)$ has the $\mathbb{P}$-compensator
\begin{align*}
  \d C_j(t) = I_j(s-)\, \lambda_j(s) \,\d s
\end{align*}
with respect to the natural filtration of $X$.
Let $$ \mathcal{T}= \{\tau \in \mathbb{T}: \tau \textrm{ is continuous} \} $$
 as in the previous example.
\begin{proposition}\label{PropositionExp3}
The ISU decomposition scheme $\delta^{ISU}$ exists  and is stable at $X$. It holds that
\begin{align}\label{Exp3D1}
\delta^{ISU}_1(X)(t) & =  \sum_{j} b_j  \int_{(0,t\wedge T]}  \frac{I_j(s)  \, \mathrm{e}^{-\int_s^T \lambda_j(u)\,\d u} }{\kappa(s)\,\mathrm{e}^{  (T-s)(r-\sigma^2) }}  \,  \,\sigma\,\d W(s)\\
 & \quad  - \alpha \sum_{i,j } b_ib_j
   \int_{(0,t\wedge T]}  \Psi_{ij}(s)  \bigg((\Upsilon_{ij}(s)-1)  \,\sigma\, \d W(s)  +  \frac{\sigma^2}{2}\d s\bigg),  \nonumber \\ \label{Exp3D2}
 \delta^{ISU}_2(X)(t) & =  \sum_{j} b_j \int_{(0,t\wedge T]}  \frac{ \mathrm{e}^{-\int_s^T \lambda_j(u)\,\d u} }{\kappa(s)\,\mathrm{e}^{  (T-s)(r-\sigma^2) }}  \, \d (N_j-C_j)(s)\\
  \nonumber & \quad +  \alpha \sum_j \int_{(0,t\wedge T]}  \big(V(s)-V(s-)\big) \d N_j(s)\\
 \nonumber  & \quad  +\alpha \sum_{i,j} b_ib_j     \int_{(0,t\wedge T]}  \Psi_{ij}(s)\,\frac{ \Upsilon_{ij}(s)( 1  + \mathds{1}_{i \neq j})  -   2 }{2} \,\d C_j(s)
\end{align}
almost surely for all $t \geq 0$, where $V(s) =  \sqrt{\mathrm{Var}[A(0)-L(0)| \sigma(X^s)
]}$ and
\begin{align*}
 \Upsilon_{ij}(s) &= \mathrm{e}^{(T-s)\sigma^2}\mathrm{e}^{\mathds{1}_{i= j}\int_s^T \lambda_j(u)\,\d u},\\  \Psi_{ij}(s) &= \frac{ I_i(s)I_j(s) \,\mathrm{e}^{-\int_s^T (\lambda_i(u)+\lambda_j(u)) \, \d u}  }{ V(s) \kappa(s)^2\, \mathrm{e}^{  2(T-s)(r-\sigma^2) } } .
\end{align*}
\end{proposition}
According to  Proposition \ref{PropositionExp3} and Theorem \ref{Prop:EquivISUandANS_F}, the formulas  \eqref{Exp3D1}\&\eqref{Exp3D2}    represent
the only decomposition of $(\varrho(X^t))_{t \geq 0}$ that is additive, normalized, and  stable.
Compared to  the decomposition formulas \eqref{D1Exp1}\&\eqref{D2Exp1} of the previous example, in  \eqref{Exp3D1}\&\eqref{Exp3D2} we have further addends that stem from  the safety margin $\alpha\,\sqrt{ \mathrm{Var}[A(0)-L(0)| \sigma(X^t)   ]}$ and that can be interpreted as the contributions of $\Phi$ and $N$ to the overall safety margin.
To our knowledge, the decomposition \eqref{Exp3D1}\&\eqref{Exp3D2} is completely  new in the literature.

\subsection*{Conservative valuation by a first-order basis}

We study an example from Jetses \& Christiansen (2021).
Let  $\d \Phi^*(t) = \phi^*(t)\,\d t$ be a technical interest intensity and $\d \Lambda^*_j(t) = \lambda_j^*(t)\, \d t$ be a technical mortality intensity that constitute a  first order basis $(\Phi^*, \Lambda_j^*)$ for the $j$-th insurance policy. The  classical conservative evaluation of life insurance liabilities uses for $A(0)-L(0)$ at time $t$ the proxy
\begin{align*}
  R(t) &= \sum_{j} p_j -  \sum_{j} b_j \frac{I_j(t)  \,\mathrm{e}^{-\int_t^T \lambda^*_j(u)\, \d u}}{\kappa(t) \,\mathrm{e}^{\int_t^T \phi^*(u)\, \d u}}  \\
   & =\sum_{j} p_j -  \sum_{j}  b_j \frac{\mathcal{E}_T( -\Lambda^*_j -(\Lambda-\Lambda^*)^t)}{\mathcal{E}_T( \Phi^* +(\Phi-\Phi^*)^t)}
\end{align*}
in case of $t \leq T$ and $R(t)=R(T)$ in case of $t>T$, see e.g.~Norberg (1999).  The representation in the second line uses the so-called Dol\'{e}ans-Dade exponential  for semimartingales,  cf.~Protter (2005, section II.8).
By setting
\begin{align*}
  X=(X_1,X_2) = (\Phi-\Phi^*, (N_j-\Lambda^*_j)_{j }),
\end{align*}
the revaluation process $R$ can be represented as
\begin{align*}
  R(t) = \varrho(X^t)
  =\sum_{j} p_j -  \sum_{j} b_j \frac{\mathcal{E}_T( -\Lambda^*_j -X_{2,j}^{t \wedge T})}{\mathcal{E}_T( \Phi^* +X_1^{t \wedge T})}.
\end{align*}
We assume that $\d \Phi = \phi(t)\, \d t$.
There exists a probability measure $\mathbb{P}^*$ such that
$$ \d C^*_j(t) = I_j(t-)\, \lambda_j^*(t)\, \d t$$
is the $\mathbb{P}^*$-compensator of $\d N_j$ with respect to the natural filtration of $N_j$.
By stopping $X_1$ and $X_2$ at asynchronous time points $t_1$ and $t_2$, we obtain the revaluation surface
\begin{align*}
  U(t_1,t_2) = \varrho((X_1^{t_1},X_2^{t_2})) = \sum_{j} p_j -  \sum_{j}b_j \frac{\mathcal{E}_T( -\Lambda^*_j -X_{2,j}^{t_2 \wedge T})}{\mathcal{E}_T( \Phi^* +X_1^{t_1 \wedge T})}.
\end{align*}
We consider all kinds of delays and  set $$\mathcal{T}= \mathbb{T}.$$
The following result can be found in Jetses \& Christiansen (2021).
\begin{proposition}
The ISU decomposition $\delta^{ISU}(X)$ in the sense of Definition \ref{DefISUdecomp} exists and almost surely equals
\begin{align}\label{Exp2D1}
\delta^{ISU}_1(X)(t) &= \sum_{j} b_j\int_{(0,t\wedge T]}\frac{I_j(s) \,\mathrm{e}^{-\int_t^T \lambda^*_j(u)\, \d u}}{\kappa(s) \,\mathrm{e}^{\int_t^T \phi^*(u)\, \d u}}   (\phi(s) -\phi^*(s)) \d s,\\
  \label{Exp2D2}
 \delta^{ISU}_2(X)(t) &= \sum_{j}b_j \int_{(0,t \wedge T]} \frac{\mathrm{e}^{-\int_t^T \lambda^*_j(u)\, \d u}}{\kappa(s) \,\mathrm{e}^{\int_t^T \phi^*(u)\, \d u}}  \big( \d N_j(s) - \d C^*_j(s)\big)
\end{align}
for all $t \geq 0$.
\end{proposition}
The proof of this proposition is more or less included in the proof of the following proposition.
\begin{proposition}\label{OwnPropExp2}
The ISU decomposition scheme $\delta^{ISU}$ exists and is stable at $X$.
\end{proposition}
According to  Proposition \ref{OwnPropExp2} and Theorem \ref{Prop:EquivISUandANS_F}, the formulas  \eqref{Exp2D1}\&\eqref{Exp2D2}   represent
the only decomposition of $(\varrho(X^t))_{t \geq 0}$ that is additive, normalized, and  stable.
The decomposition \eqref{Exp2D1}\&\eqref{Exp2D2} is popular in life insurance, cf.~Norberg (1999).  While the actuarial literature derives it in a  heuristic way, we rediscover it here by axiomatic principles. 

\section{Proofs of the results of section \ref{SectionExamples}}

\begin{proof}[Proof of Proposition \ref{LemmaExp1}] For $t > T$ it suffices to notice that  $t \mapsto \varrho(X^t)$ is constant. So in the remaining proof we focus on the time interval $[0,T]$ only.
As $\Phi$ and $N$ are independent under $\mathbb{Q}$ and since $\mathbb{Q}( N_j(T)=0 | N_j(t)=0)= \exp\{-\int_t^T \lambda^{\mathbb{Q}}_j(s) \d s \}$, we can conclude that
\begin{align*}
  \E^{\mathbb{Q}} \big[A(0)-L(0) \big|\, \Phi , N^t \big]&=  \sum_{j} p_j -  \sum_{j} b_j \frac{1}{\kappa(T)}  \E^{\mathbb{Q}} [1-N_j(T) | N_j(t)=0]\,I_j(t)\\
  &=\sum_{j} p_j -  \sum_{j} b_j \frac{1}{\kappa(T)}  \mathrm{e}^{-\int_t^T \lambda^{\mathbb{Q}}_j(s) \d s }\,I_j(t)
\end{align*}
for  $t \leq T$.  Moreover, by applying the tower property of conditional expectations and using again the  independence of $\Phi$ and $N$ under $\mathbb{Q}$, we can show that  the revaluation surface equals
\begin{align*}
  U(t_1,t_2) &= \E^{\mathbb{Q}} \big[A(0)-L(0)\, \big|\, \sigma(\Phi^{t_1} , N^{t_2})
  \big]\\
  &=\sum_{j} p_j -  \sum_{j} b_j \frac{\mathrm{e}^{ - (T-t_1)(\mu-\sigma^2) }}{\kappa(t_1)}  \mathrm{e}^{-\int_{t_2}^T \lambda^{\mathbb{Q}}_j(s) \d s }\, I_j(t_2)
\end{align*}
for $t_1, t_2 \leq T$.
Since $1/\kappa(t) = \exp\{-t \mu + t \sigma^2/2 - \sigma W^{\mathbb{Q}}(t)\}$,
by applying It\^{o}'s formula we can show that
\begin{align*}
  U(t_{k+1},t_k)-U(t_{k},t_k) =    \sum_{j} b_j I_j(t_k) \, \mathrm{e}^{-\int_{t_k}^T \lambda^{\mathbb{Q}}_j(s) \d s } \,\mathrm{e}^{ - T(\mu-\sigma^2) } \int_{(t_k,t_{k+1}]} \mathrm{e}^{ - s\, \sigma^2 /2 -  \sigma W^{\mathbb{Q}}(s)}\, \sigma\, \d W^{\mathbb{Q}}(s)
\end{align*}
for $t_k < t_{k+1}\leq T$. By applying the dominated convergence theorem for It\^{o} integrals we obtain the ISU decomposition \eqref{D1Exp1}.
Similarly, we can show that
\begin{align*}
  U(t_{k+1},t_{k+1})-U(t_{k+1},t_k) =     \sum_{j} b_j \frac{1}{\kappa(t_{k+1})}  \,\mathrm{e}^{ - (T-t_{k+1})(\mu-\sigma^2) } \int_{(t_k,t_{k+1}]}  \mathrm{e}^{-\int_s^T \lambda^{\mathbb{Q}}_j(s)}  \d (N_j-C_j)(s),
\end{align*}
which converges almost surely to \eqref{D2Exp1} according to the dominated convergence theorem for Lebesgue integrals.

For continuous delays $\tau$, analogously to above  we can show that $X \diamond \tau$ generates the mortality surface
\begin{align*}
  U^{\tau}(t_1,t_2) =\sum_{j} p_j -  \sum_{j} b_j \frac{\mathrm{e}^{ - (T-\tau_1(t_1))(\mu -\sigma^2) }}{\kappa(\tau_1(t_1))}  \mathrm{e}^{-\int_{\tau_2(t_2)}^T \lambda^{\mathbb{Q}}_j(u) \,\d u }\, I_j(\tau_2(t))
\end{align*}
for any time points $t_1, t_2 $ such that $\tau_1(t_1), \tau_2( t_2) \leq T$. Building the ISU limits with the help of the dominated convergence theorems leads to the decomposition
\begin{align}\label{D1Exp1_b}
   &\sum_{j} b_j  \int_{(0,\tau_1(t) \wedge T]}  \frac{\mathrm{e}^{ - (T-s)(\mu-\sigma^2)}}{\kappa(s)}   \, I_j\circ \tau_2\circ \tau_1^{-1}(s)  \, \mathrm{e}^{-\int_{\tau_2\circ \tau_1^{-1}(s)}^T \lambda^{\mathbb{Q}}_j(u)\, \d u} \sigma\,\d W^{\mathbb{Q}}(s),\\
  \label{D2Exp1_b}
   &\sum_{j} b_j \int_{(0,\tau_2(t) \wedge T]}  \frac{\mathrm{e}^{ - (T-\tau_1\circ \tau_2^{-1}(s))(\mu-\sigma^2) }}{\kappa\circ \tau_1 \circ \tau_2^{-1}(s)}\,   \mathrm{e}^{-\int_s^T \lambda^{\mathbb{Q}}_j(u)\, \d u}  \, \d (N_j-C_j)(s),
\end{align}
for $\tau_i^{-1}(s) := \inf\{ u \geq 0 : \tau_i(u)\geq s\}$.
The stochastic integral \eqref{D1Exp1_b} is well-defined  if we read it as an It\^{o} integral  with respect to the natural filtration of $X \diamond \tau$.
If we replace $\tau$ in \eqref{D1Exp1_b} and  \eqref{D2Exp1_b} by a refining sequence       $(\tau^n)_{n \in \mathbb{N}} \subset \mathcal{T}$   of  continuous delays  that increase to identity, then for $n \rightarrow \infty$ Proposition \ref{tauIncresToInf}(a) and the dominated convergence theorems yield  that the left-continuous versions of \eqref{D1Exp1_b} and  \eqref{D2Exp1_b}  converge in probability to the left-continuous versions of  \eqref{D1Exp1} and  \eqref{D2Exp1}.
\end{proof}

\begin{proof}[Proof of Proposition \ref{PropositionExp3}]
For $t > T$ it suffices to notice that  $t \mapsto \varrho(X^t)$ is constant. So in the remaining proof we focus on the time interval $[0,T]$ only. We may calculate the expectation part and the standard deviation part of $\varrho $  separately and add the two components afterwards. For the expectation part we just need to follow the proof of Proposition \ref{LemmaExp1}.  So the remaining proof deals with the decomposition of the standard deviation part only.  Let $\widetilde{U}$ be the revaluation surface  of the standard deviation part,
\begin{align*}
 \widetilde{U}(t_1,t_2) = \alpha \sqrt{\mathrm{Var}[A(0)-L(0)| \sigma(\Phi^{t_1} , N^{t_2})
 ] }.
\end{align*}
Let $K(s) = \exp\{  (T-s)(r-\sigma^2) \}$ and $q_j(s) = \exp\{-\int_s^T \lambda_j(u) \, \d u\}$.
Analogously to the proof of Proposition \ref{LemmaExp1} we can show that
\begin{align*}
\E \Big[  \sum_j b_j \frac{I_j(T)}{\kappa(T)} \Big|\, \Phi , N^t \Big]
  &=\sum_{j} b_j  \frac{1}{\kappa(T)} q_j(t)\,I_j(t),\\
  \E \Big[\Big(  \sum_j b_j \frac{I_j(T)}{\kappa(T)} \Big)^2 \Big|\, \Phi , N^t \Big]
  &=\sum_{j } b_j^2   \frac{1}{\kappa(T)^2} q_j(t)\,I_j(t)
  + \sum_{i,j:i \neq j } b_i b_j  \frac{1}{\kappa(T)^2} q_i(t)q_j(t)\,I_i(t)\,I_j(t),
\end{align*}
where we used that fact that $N_1, N_2, \ldots$ are independent.
So,  by applying the tower property of conditional expectations and using the  independence of $\Phi$ and $N$, we obtain that
\begin{align*}
  & \mathrm{Var}[A(0)-L(0)| \sigma(\Phi^{t_1} , N^{t_2}) 
   ] \\
  &=   \sum_{i,j } b_ib_j   \frac{K(t_1)^2}{\kappa(t_1)^2} q_i(t_2)\, \,I_i(t_2)\,I_j(t_2)\Big( \mathrm{e}^{(T-t)\sigma^2}\,q_j(t_2)^{\mathds{1}_{i\neq j}}-q_j(t_2)\Big),
\end{align*}
since  $  1/\kappa(T)^2   = \exp\{-2 (T-t) r+ 2(T-t) \sigma^2/2 - 2 \sigma (W(T)-W(t))\} /  \kappa(t_1)^2$.
By applying It\^{o}'s formula,  we get
\begin{align*}
   &\widetilde{U}(t_{k+1},t_{k} )- \widetilde{U}(t_{k},t_k )\\
   &=    -  \alpha \sum_{i,j } b_ib_j    q_i(t_k)\,I_i(t_k)\,I_j(t_k)\\
   &\qquad \times
   \int_{(t_k,t_{k+1}]}  \frac{\alpha}{ 2\widetilde{U}(s,t_{k})}  \frac{K(s)^2}{\kappa(s)^2} \bigg(\Big( \mathrm{e}^{(T-s)\sigma^2}\,q_j(t_k)^{\mathds{1}_{i\neq j}}-q_j(t_k)\Big)  2 \sigma \d W(s)  + q_j(t_k) \sigma^2 \d s\bigg)
\end{align*}
  and
\begin{align*}
   &\widetilde{U}(t_{k+1},t_{k+1} )- \widetilde{U}(t_{k+1},t_k )\\
  &=   \alpha \sum_{i,j} b_i b_j   \frac{K(t_{k+1})^2}{\kappa(t_{k+1})^2}  \int_{(t_k,t_{k+1}]} \alpha \frac{ I_i(s)I_j(s)\, q_i(s) }{ 2\,\widetilde{U}(t_{k+1},s)}  \\
  & \qquad \qquad \times  \Big( \mathrm{e}^{(T-t_{k+1})\sigma^2}(  \lambda_i(s) + \mathds{1}_{i \neq j} \lambda_j(s))q_j(s)^{\mathds{1}_{i\neq j}}     -(  \lambda_i(s) + \lambda_j(s)) q_j(s)\Big) \,\d s\\
  & \quad +   \sum_j \int_{(t_k,t_{k+1}]} \big(\widetilde{U}(t_{k+1},s)-\widetilde{U}(t_{k+1},s-)\big) \d N_j(s),
\end{align*}
using the fact that $N_1, N_2, \ldots $ almost surely have no simultaneous jumps.
By applying the dominated convergence theorem and rearranging the terms, we arrive at  the ISU decomposition  according to \eqref{Exp3D1} and \eqref{Exp3D2}.
Likewise, one can show $X \diamond \tau$ has the ISU decomposition
\begin{align*}
 \delta_1^{ISU}(X \diamond \tau)(t) & =  \sum_{j} b_j  \int_{(0,\tau_1(t)\wedge T]}  \frac{I_j\circ \sigma_{21}(s)  \, q_j\circ \sigma_{21}(s)}{ \kappa(s) K(s)}   \,  \,\sigma\,\d W(s)\\
 & \quad  - \alpha \sum_{i,j } b_ib_j
   \int_{(0,\tau_1(t)\wedge T]} \alpha  \frac{ I_i\circ\sigma_{21}(s)\,I_j\circ\sigma_{21}(s) \, q_i\circ\sigma_{21}(s)\,q_j\circ\sigma_{21}(s)  }{\kappa(s)^2K(s)^2  \widetilde{U}(s,\sigma_{21}(s)) } \\
   & \qquad \qquad\qquad \qquad  \times \bigg(\bigg(\frac{\mathrm{e}^{(T-s)\sigma^2}}{q_j\circ \sigma_{21}(s)^{\mathds{1}_{i= j}}}-1\bigg) \sigma \d W(s)  +  \frac{\sigma^2}{2}\d s\bigg),  \\
 \delta_2^{ISU}(X \diamond \tau)(t) & =  \sum_{j} b_j \int_{(0,\tau_2(t)\wedge T]}  \frac{ q_j(s)}{\kappa\circ \sigma_{12}(s)\, K\circ \sigma_{12}(s)}\,  \, \d (N_j-C_j)(s)\\
   & \quad +  \sum_j \int_{(0,\tau_2(t)\wedge T]}  \big(\widetilde{U}(\sigma_{12}(s),s)-\widetilde{U}V(\sigma_{12}(s),s-)\big) \d N_j(s)\\
  & \quad  +\alpha \sum_{i,j} b_i b_j     \int_{(0,\tau_2(t)\wedge T]}  \alpha \frac{I_i(s)I_j(s) q_i(s)q_j(s) }{ \kappa\circ \sigma_{12}(s)^2K\circ \sigma_{12}(s)^2  \widetilde{U}(\sigma_{12}(s),s)} \\
   & \qquad \qquad\qquad \qquad  \times \frac{1}{2}
   \bigg(  \bigg(\frac{\mathrm{e}^{(T-\sigma_{12}(s))\sigma^2}}{q_j(s)^{\mathds{1}_{i= j}}}-1\bigg)(  \lambda_i(s) + \mathds{1}_{i \neq j} \lambda_j(s))  -   \mathds{1}_{i = j} \lambda_j(s)\bigg) \d s
\end{align*}
for $\sigma_{12}(t) := \tau_1\circ \tau_2^{-1}$ and $\sigma_{21}:=\tau_2\circ \tau_1^{-1}$ where $\tau_i^{-1}(s) := \inf\{ u \geq 0 : \tau_i(u)\geq s\}$. If we replace $\tau$  by a refining sequence       $(\tau^n)_{n \in \mathbb{N}} \subset \mathcal{T}^c$   of  continuous delays  that increase to identity, then for $n \rightarrow \infty$ Proposition \ref{tauIncresToInf}(a) and the dominated convergence theorem verify the stability of $\delta^{ISU}$  at $X$.
\end{proof}

\begin{proof}[Proof  of Proposition \ref{OwnPropExp2}]
For $t > T$ it suffices to notice that  $t \mapsto \varrho(X^t)$ is constant. So in the remaining proof we focus on the time interval $[0,T]$ only.
The delayed risk basis $X \diamond \tau$ has the revaluation surface
\begin{align*}
  U^{\tau} (t_1,t_2) =  \sum_{j} p_j -  \sum_{j} b_j I_j\circ \tau_2(t_2) \, \mathrm{e}^{-\int_0^{\tau_2(t_2)} \phi(u)\, \d u}  \mathrm{e}^{-\int_{\tau_1(t_1)}^T \phi^*(u)\, \d u} \mathrm{e}^{-\int_{\tau_2(t_2)}^T \lambda^*_j(u)\, \d u}
\end{align*}
for all $t_1,t_2 $ with $\tau_1(t_1), \tau_2(t_2) \leq T$.
By applying It\^{o}'s formula  we can show that
\begin{align*}
   &U^{\tau}(t_{k+1},t_k )- U^{\tau}(t_{k},t_k )\\
     &=  \sum_j \int_{(\tau_1(t_k),\tau_1(t_{k+1})]} b_j\frac{I_j\circ \tau_2(t_k)}{\kappa(s)} \mathrm{e}^{-\int_t^T \phi^*(u)\, \d u}  \mathrm{e}^{-\int_{ \tau_2(t_k)}^T \lambda^*_j(u)\, \d u}(\phi(s) -\phi^*(s)) \d s.
\end{align*}
Building the ISU limits with the help of the dominated convergence theorem leads to the decomposition
\begin{align*}
   \delta_1^{ISU}(X \diamond \tau)(t)
     &=  \sum_j \int_{(0,\tau_1(t) \wedge T]} b_j\frac{I_j\circ \tau_2\circ \tau_1^{-1}(s)}{\kappa(s)} \mathrm{e}^{-\int_s^T \phi^*(u)\, \d u}  \mathrm{e}^{-\int_{ \tau_2 \circ \tau_1^{-1}(s)}^T \lambda^*_j(u)\, \d u}(\phi(s) -\phi^*(s)) \d s
\end{align*}
for $\tau_i^{-1}(s) := \inf\{ u \geq 0 : \tau_i(u)\geq s\}$. Likewise we can show that
\begin{align*}
   \delta_2^{ISU}(X \diamond \tau)(t)
     &=  \sum_j \int_{(0,\tau_2(t) \wedge T]} b_j\frac{\mathrm{e}^{-\int_{\tau_2\circ \tau_1^{-1}(s)}^T \phi^*(u)\, \d u}}{\kappa\circ \tau_1\circ \tau_2^{-1}(s)} \mathrm{e}^{-\int_{ s}^T \lambda^*_j(u)\, \d u} \big( \d N_j(s) - \d C^*_j(s)\, \d s\big).
\end{align*}
For $t_1,t_2 $ with $\tau_1(t_1), \tau_2(t_2) \leq T$ the process $t \mapsto \varrho ( (X \diamond \tau)^t$ is constant.
If we replace $\tau$  by a refining sequence       $(\tau^n)_{n \in \mathbb{N}} $   of   delays  that increase to identity, then for $n \rightarrow \infty$ Proposition \ref{tauIncresToInf}(a) and the dominated convergence theorem verify the stability of $\delta^{ISU}$  at $X$.
\end{proof}

\section{Conclusion}

In insurance practice, the currently preferred method for the change analysis of an insurer's own funds is the SU decomposition principle with yearly time steps. 
 The main disadvantage of the SU decomposition principle is its dependence on the update order of the risk factors. This observation raises the question whether there exist more convincing decomposition concepts.
In this paper we started the search for better decomposition methods from an axiomatic basis that postulates additivity, normalization and stability as desirable properties of profit and loss decompositions. Our key finding is that these three axioms  necessarily lead to ISU decompositions, provided that the latter exist and are stable themselves.
The definition of ISU decompositions implies that they can be approximated by SU decompositions. So, insurers can easily implement approximative ISU decompositions in practice.

For each of the three examples in section \ref{SectionExamples}  we verified the existence and stability of ISU decompositions  separately. This is due to the fact that our examples for  the mapping $\varrho$ are mathematically rather different. Finding general  existence and stability results is a challenge for future research. Our three examples indicate that our axiomatic concept might work for a large variety of revaluation mappings $\varrho$. So this paper is rather the starting point of a research journey than its end point.

\section*{References}

\bigskip {\small
\begin{list}{}{\leftmargin1cm\itemindent-1cm\itemsep0cm}
\item{Bashford, T., Dubischar, D., 2020. Die Ver\"{a}nderungsanalyse (AoC) in
einem konsistenten Ansatz. Presentation at Herbsttagung 2020 of the German Actuarial Association, available on www.actuview.com.}

\item{Biewen, M., 2014. A general decomposition formula with interaction effects. Applied Economics Letters, 21(9), 636-642.}

\item{Candland, A., Lotz, C., 2014. Profit and loss attribution. In: Internal models and Solvency II - From regulation to implementation, edited by Paolo Cadoni. Risk Books. }

\item{ CFO Forum, 2016.  Presentation of analysis of earnings. Market Consistent Embedded Value Principles, Appendix B.}

\item{European Union, 2015. Excess of Assets over Liabilities — explained by technical provisions.  Official Journal of the European Union, L 347/386.}

\item{F\"{o}llmer, H.~ and Schied, A., 2016. Stochastic Finance. De Gruyter, 4th edition.}

\item{Fortin, N., Lemieux, T., Firpo, S., 2011. Chapter 1-decomposition methods in economics. Volume 4, Part A of Handbook of Labor Economics. Elsevier 10, S0169-7218.}

\item{International Accounting Standards Board, 2017. Movements in insurance contract liabilities analysed by components.  IFRS 17 Insurance Contracts, Table 3.}

\item{Jetses, J.~and Christiansen, M., 2021.  A General Surplus Decomposition Principle in
Life Insurance. arXiv:2111.12967v1.}

\item{Norberg, R., 1999. A theory of bonus in life insurance. Finance and Stochastics 3/4, 373-390.}

\item{Schilling, K., Bauer, D., Christiansen, M.C., Kling, A., 2020. Decomposing Dynamic Risks into Risk Components. Management Science 66/12, 5485-6064.}

\item{Sheldon, T.J.~and Smith, A.D., 2004. Market consistent valuation of life assurance business, British Actuarial Journal 10(3), 543-626.}

\item{Shorrocks, A.F., 2013. Decomposition procedures for distributional analysis: a unified framework based on the Shapley value. Journal of Economic Inequality 11(1), 99-126.}

\end{list}}

\end{document}